\documentclass[conference]{IEEEtran}
\IEEEoverridecommandlockouts
\usepackage{cite}
\usepackage{amsmath,amssymb,amsfonts, amsthm}
\usepackage{algorithmic}
\usepackage{graphicx}
\usepackage{textcomp}
\usepackage{xcolor}
\usepackage{mathtools}
\usepackage{subfig}
\usepackage{url}
\def\BibTeX{{\rm B\kern-.05em{\sc i\kern-.025em b}\kern-.08em
    T\kern-.1667em\lower.7ex\hbox{E}\kern-.125emX}}

\newcommand{\x}{\textnormal{\textbf{x}}}
\renewcommand{\v}{\textnormal{\textbf{v}}}
\newcommand{\p}{\textnormal{\textbf{p}}}

\newcommand{\voi}{\textnormal{\textbf{v}}_i}
\newcommand{\poi}{\textnormal{\textbf{p}}_i}
\newcommand{\ri}{\textnormal{\textbf{r}}_i}
\renewcommand{\u}{\textnormal{\textbf{u}}}
\newcommand{\K}{\mathcal{K}}
\newcommand{\Set}{\mathcal{S}}

\DeclareMathOperator{\Int}{Int}
\DeclareMathOperator{\atan}{atan2}
\newcommand{\vri}{\textnormal{\textbf{v}}_{\mathrm{r},i}}

\newtheoremstyle{idefinition}
{}                
{}                
{}        
{}                
{\bfseries\itshape}       
{.}               
{ }               
{}                
\makeatletter
\def\thmhead@plain#1#2#3{%
  \thmname{#1}\thmnumber{\@ifnotempty{#1}{ }\@upn{#2}}%
  \thmnote{ {\the\thm@notefont(\textit{#3})}}}
\let\thmhead\thmhead@plain
\makeatother
\theoremstyle{definition}
\newtheorem{remark}{Remark}
\newtheorem{assumption}{Assumption}
\theoremstyle{idefinition}
\newtheorem{defi}{Definition}

\theoremstyle{theorem}
\newtheorem{lemma}{Lemma}
\newtheorem{theorem}{Theorem}

\begin{document}

\title{Safety-Critical Control of Nonholonomic Vehicles in Dynamic Environments using Velocity Obstacles
}

\author{Aurora Haraldsen$^{1}$, Martin S. Wiig$^{2}$, Aaron D. Ames$^{3}$, Kristin Y. Pettersen$^{1,2}$
\thanks{*This work was supported by the Research Council of Norway through project No. 302435 and the Centres of Excellence funding scheme, project No. 223254 – NTNU AMOS.
This research was supported in part by the National Science Foundation (CPS Award \#1932091) and Nodein Inc.}
\thanks{$^{1}$Aurora Haraldsen and Kristin Y. Pettersen are with the Department of Engineering Cybernetics, Norwegian University of Science and Technology, NO-7491 Trondheim, Norway.  {\{\tt\small Aurora.Haraldsen, Kristin.Y.Pettersen\}@ntnu.no}}%
\thanks{$^{2}$Martin S. Wiig and  Kristin Y. Pettersen are with the Norwegian Defence Research Establishment, P.O. Box 25, N-2027 Kjeller, Norway.
        {\tt\small Martin-Syre.Wiig@ffi.no}}%
\thanks{$^{3}$Aaron D. Ames is with the Department of Mechanical
and Civil Engineering, California Institute of Technology, Pasadena,
CA 91125, USA. {\tt\small ames@caltech.edu}}
        }

\maketitle

\begin{abstract}
This paper considers collision avoidance for vehicles with first-order nonholonomic constraints maintaining nonzero forward speeds,  moving within dynamic environments. We leverage the concept of control barrier functions (CBFs) to synthesize control inputs that prioritize safety, where the safety criteria are derived from the \emph{velocity obstacle principle}. Existing instantiations of CBFs for collision avoidance, e.g., based on maintaining a minimal distance, can result in control inputs that make the vehicle stop or even reverse. The proposed formulation effectively separates speed control from steering, allowing the vehicle to maintain a forward motion without compromising safety. This is beneficial for ensuring that the vehicle advances towards its desired destination, and it is moreover an underlying requirement for certain vehicles such as marine vessels and fixed-wing UAVs. Theoretical safety guarantees are provided, and numerical simulations demonstrate the efficiency of the strategy in environments containing moving obstacles. 
\end{abstract}

\section{Introduction}

With the increasing use of autonomous vehicles for commercial and scientific purposes, ensuring safety formally and in practice has become ever more pressing.  Within robot navigation, an important aspect of safety involves avoiding collision with obstacles that appear in the robot's surroundings. In dynamic environments, i.e., containing moving obstacles, the collision avoidance system must accommodate an unpredictable information picture providing only a limited time to react to a collision. As a result,  the effectiveness of planning-based algorithms is reduced significantly, highlighting the need for developing real-time reactive methods that ensure safety. 

Safety can in many applications be framed as the property of always remaining in a safe set, i.e., forward set invariance.  Control barrier functions~(CBFs)~\cite{Ames17} present a means to ensure safety of dynamical systems, by providing a framework to \textit{enforce} the forward invariance property. The controller synthesis is usually framed as an optimization problem in which the deviation between the nominal system input and safe inputs is minimized~\cite{Ames2019}, allowing for reactive adjustments of the vehicle's control inputs in a minimally intrusive manner.  Collision avoidance has been studied within the framework of CBFs, where most previous work consider static environments and multi-agent systems~\cite{Wang17,Ames2019,Glotfelter19, BORRMANN201568, Basso2020, Marley2021}. For double integrator systems, the approach~\cite{Wang17} can be applied to avoid a dynamic obstacle by viewing it as an agent with zero control inputs. Safety can thus only be concluded if the obstacle moves at a constant velocity. The same constant velocity assumption is made in \cite{Thyri2020}, where CBFs are applied for the control of a fully actuated ASV. Importantly, these are both holonomic systems where CBFs have found great success.

This study specifically considers vehicles with unicycle-kinematics, constituting a first-order nonholonomic constraint. Since the rotation rate does not appear through differentiation of the position, the result is limited control authority if the barrier function is defined solely from the position of the vehicle. To bypass this limitation, collision avoidance may be assured for a point in front the vehicle~\cite{robotarium,Glotfelter19}. However, this approach has not been studied for dynamic environments in which the perturbation of the vehicle center could lead to unwanted behaviour. Noticing that the control inputs appear through a second differentiation of the position, the approach taken by \cite{Molnar23} uses a high-order barrier function~\cite{Xiao22} for controlling a segway to avoid collision. Although a dynamic obstacle is considered, the segway is restricted to move on a straight line, thus making it a holonomic system.

In the context of applying CBFs to UAVs and UGVs, \cite{Tayal2023} and \cite{Thontepu2023} formulate a barrier function based on the collision cone concept. A collision cone refers to the set of  relative velocities between the vehicle and an obstacle resulting in a future collision~\cite{collisioncones,Chakravarthy1998,FioriniShiller1998}, thus offering a predictive notion of safety that can be exploited to keep the vehicle from colliding with moving obstacles. The main limitation of this approach is, however, that the barrier function is not defined when the relative velocity vector is zero.  The equivalent velocity obstacle (VO) is perhaps more suited for control design, as it relates an obstacle to the absolute velocities of the vehicle~\cite{FioriniShiller1998}. The concept has been applied to nonholonomic systems, but never in the context of CBFs. Moreover, \cite{Wilkie09} and \cite{VandenBerg2011} adapts the VO to vehicles with car-like kinematics, however, the control inputs are only guaranteed to be safe so long as the obstacles continue along their current paths.  VO-based strategies have been proposed for marine surface vessels~\cite{Kuwata2014, Huang2018, Cho2019}, underwater vehicles~\cite{Zhang2017}, and UAVs~\cite{Jenie2016},
but safety guarantees are not provided.  While the algorithm~\cite{Haraldsen2020} is shown to guarantee safety for unicycle-type vehicles under explicit conditions, the speed of the vehicle is assumed to be constant.

Inspired by the concept of velocity cones, we propose a novel strategy using CBFs for nonholonomic vehicles where it is beneficial to maintain a positive speed while avoiding obstacles. The basis of the safety condition is drawn from the velocity obstacle cone rather than the collision cone, thus circumventing the issues related to the latter formulation. As opposed to most other collision avoidance strategies formed on CBFs, we decouple the control of the speed from the steering.  An inherent advantage of this formulation is that the avoidance strategy is less invasive with respect to the nominal vehicle speed, which makes it better suited for vehicles with speed requirements, such as underactuated marine vehicles and fixed-wing UAVs. While more conventional CBFs result in vanishing control authority when the vehicle must maintain a nonzero forward speed~\cite{Marley2021},  the proposed strategy preserves the relative degree at all times, thus maintaining safety of the vehicle in any configuration. Safety guarantees are analytically derived, and  the resulting performance is demonstrated in simulations with multiple dynamic obstacles.

The paper is organized as follows. Section~\ref{sec:MathFramework} provides a brief introduction to CBFs. Section~\ref{sec:Problem}  details the problem and presents the velocity obstacle as a means to encode collisions between moving objects in the velocity space. The concept is leveraged in Section~\ref{sec:Controller} to formulate CBFs that modify the inputs of the vehicle to preserve safety, where theoretical guarantees of safety are provided. Simulations are presented in Section~\ref{sec:Simulations}, and concluding remarks are given in Section~\ref{sec:Conclusion}.

\section{Mathematical Framework}\label{sec:MathFramework}
The control strategy proposed in this paper builds upon the notion of barrier functions. This section establishes the notation and theoretical foundation of this framework. Moreover, it briefly discusses controller synthesis using barrier functions. For the intuition behind, and the mathematical proofs of the theory presented here, the reader is referred to \cite{Ames17,XU15}.

\begin{defi}[Extended class-$\K$ function]
A continuous function $\alpha : \left(-c,d\right)\rightarrow \left(-c,\infty\right)$ for $c,d > 0$ is said to belong to extended class-$\mathcal{K}$, denoted as $\alpha \in \K_e$, if it is strictly increasing and $\alpha(0)=0$.
\end{defi}
We consider a nonlinear system
\begin{equation}\label{eq:nonlinear_eq}
    \Dot{\x} = f(\x),
\end{equation}
for $\x \in \mathcal{D}$, where $\mathcal{D} \subseteq \mathbb{R}^n$ is open and connected, and the mapping $f : \mathcal{D} \rightarrow \mathbb{R}^n$ is locally Lipschitz continuous on $\mathcal{D}$. Given an initial condition $\x(t_0) \in \mathcal{D}$, let $I(\x(t_0)) = [t_0, \tau_{\max})$ denote the maximum interval of existence for the unique solution $\x(t)$ of \eqref{eq:nonlinear_eq}. 
\begin{defi}
    A set $\Set\subseteq \mathcal{D}$ is forward invariant w.r.t.~\eqref{eq:nonlinear_eq} if 
    \begin{equation}
        \x(t_0) \in \Set \implies \x(t) \in \Set,\qquad \forall t \in I(\x(t_0)).
    \end{equation}
\end{defi} 
\begin{defi}\label{def:S}
    A set $\Set\subseteq \mathcal{D}$ is a 0-superlevel set of a continuously differentiable function $h: \mathcal{D} \rightarrow \mathbb{R}$ when
    \begin{subequations}\label{eq:S}
      \begin{align}
          \Set &= \{\x\in \mathcal{D} : h(\x) \geq 0\},\\
          \partial \Set &= \{\x\in \mathcal{D} : h(\x) = 0\},\\
          \Int(\Set) &= \{\x\in \mathcal{D} : h(\x) > 0\}.
      \end{align}
    \end{subequations}
\end{defi} 
Safety can in many contexts be defined as the property of evolving inside a $0$-superlevel set.  Accordingly, the control objective reduces to ensuring forward invariance of the set. A convenient condition for this property is given as follows.
\begin{lemma}\label{lem:invariance}
Given a set $\Set$ as in~\eqref{eq:S}, if
\begin{equation}\label{eq:hdot_geq_alpha}
    \dot{h}(\x) \geq \alpha(h(\x)), \qquad \forall \x \in \mathcal{D}, \quad \alpha \in \K_e,  
\end{equation}
then $\mathcal{S}$ is forward invariant w.r.t. \eqref{eq:nonlinear_eq}.
\end{lemma}
 A function $h$ satisfying \eqref{eq:hdot_geq_alpha} is commonly referred to as a \textit{barrier function} for \eqref{eq:nonlinear_eq}.

Consider now~\eqref{eq:nonlinear_eq} in control affine form
\begin{equation}\label{eq:nonlinear_eq_u}
    \Dot{\x} = f(\x) + g(\x)\u,\\
\end{equation}
with $\x \in \mathcal{D}$ and $f$ defined as before, $g : \mathcal{D} \rightarrow \mathbb{R}^{n\times m}$ is locally Lipschitz continuous on $\mathcal{D}$, and $\u \in \mathcal{U} \subseteq \mathbb{R}^m$ is Lipschitz continuous, where $\mathcal{U}$ is the set of admissible inputs. We desire a control input $\u$ that forces the system to evolve inside some safe set $\mathcal{S}$. This motivates the introduction of \textit{control} barrier functions (CBFs): 
\begin{defi}\label{def:CBF}
    Given a set $\Set$ as in~\eqref{eq:S},  $h$ is a CBF for~\eqref{eq:nonlinear_eq_u} if there exists an extended class-$\K$ function $\alpha$ such that
    \begin{equation}
       \sup_{\u\in \mathcal{U}} \,\underbrace{\nabla  h(\x)( f(\x) + g(\x) \u)}_{\dot{h}(\x,\u)}+ \alpha(h(\x)) \geq 0, \qquad \forall \x \in \mathcal{S}.
    \end{equation}
\end{defi}

\begin{theorem}\label{thm:CBF}
If $h$ is a CBF for \eqref{eq:nonlinear_eq_u} associated with a set $\mathcal{S}$~\eqref{eq:S}, then any Lipschitz continuous controller $ \u = k(\x)$ satisfying
\begin{equation}\label{eq:safe_constraint}
  \nabla  h(\x)( f(\x) + g(\x)k(\x))+ \alpha(h(\x)) \geq 0, \qquad \forall \x \in \mathcal{S},
\end{equation}
renders $\mathcal{S}$ forward invariant, i.e., it ensures $\x(t_0) \in \mathcal{S} \implies \x(t) \in \mathcal{S},\,\, \forall t \in I(\x(t_0))$.
\end{theorem}
Since the constraint~\eqref{eq:safe_constraint} is affine in $\u$, it can be incorporated into a quadratic optimization problem: Given a desired control input $k_{\mathrm{d}}:\mathcal{D} \rightarrow \mathcal{U}$, we synthesize a \textit{safe} control input $k : \mathcal{D} \rightarrow \mathcal{U}$ by solving 
\begin{equation}
\begin{split}\label{eq:QP}
     k(\x) = \underset{\u \in \mathcal{U}}{\arg \mathrm{min}}& \,\,\Vert \u-k_{\mathrm{d}}(\x)\Vert^2\\
     \text{subj. to} & \,\,\dot{h}(\x,\u) \geq -\alpha(h(\x)),
\end{split}
\end{equation}
in which the nominal control input is modified in a point-wise optimal fashion. The framework is extended to systems of arbitrary relative degree in~\cite{Xiao22}, while~\cite{Glotfelter17} introduces nonsmooth barrier functions relaxing the requirement of continuous differentiability. Environmental CBFs~\cite{Molnar23} present a generalization to systems acting upon a time-varying environment. To conserve space, a function will in the following parts be denoted without explicit reference to its inputs,  thereby differing from the notation used in this section.

\section{Problem and Motivation}\label{sec:Problem}
We consider a nonholonomic vehicle described by
\begin{equation}\label{eq:vehicle}
    \dot{x} = v \cos\psi,\quad
    \dot{y} = v \sin \psi,\quad
    \dot{\psi} = u_1,\quad
    \dot{v}   = u_2,
\end{equation}
where $\p\triangleq[x,y]^\top$ are the Cartesian coordinates of the vehicle, $\v \triangleq \dot{\p}$ is the linear velocity, $\psi \in (-\pi,\pi]$ is the orientation, and $v \in \mathbb{R}$ is the speed. The turning rate $r \triangleq \dot{\psi}$ and acceleration $a \triangleq \dot{v}$ are controlled through the inputs $\u \triangleq [u_1,u_2]^\top \in \mathcal{U}$. The space of admissible inputs is defined by \begin{equation}
    \mathcal{U} = \{\u \in \mathbb{R}^2 :  |r| \leq r_{\max}, |a| \leq a_{\max}\},
\end{equation}
where $r_{\max},a_{\max}>0$ represent the maximum rotation rate and acceleration, respectively. We remark that this system can be expressed in the control affine form~\eqref{eq:nonlinear_eq_u} with $\x \triangleq [\p^\top, \psi, v]^\top \in \mathcal{D}$ and $g = \begin{bmatrix} \textbf{0} & \textbf{I}
\end{bmatrix}$, where $\textbf{0}, \textbf{I}\in \mathbb{R}^{2\times2}$ denote the zero and identity matrices, respectively.

In this study, we consider vehicles that not only are nonholonomic, but additionally need to maintain nonzero and bounded forward speeds:
\begin{assumption}\label{ass:v}
    The vehicle speed is bounded by
    \begin{equation}
        v \in (0,v_{\max}].
    \end{equation}
\end{assumption}
This assumption implies that the vehicle cannot stop or move backwards to avoid collision, which applies to a variety of real systems. While for some vehicles it may not be strictly necessary to maintain a positive speed, preventing reversal can be advantageous because the design of the vehicle shape or its propulsion system makes it more energy efficient to move forwards than in reverse. In addition, forward motion can serve as a means to facilitate future progress towards a desired destination.

\subsection{Control Objectives}
We consider the case where the vehicle nominally moves towards a target position $\p_\mathrm{t} \in \mathbb{R}^2$. Note, however, that a wide range of nominal behaviors are allowed within the proposed approach. The task is completed once the vehicle reaches within an acceptable distance, $d_{\mathrm{acc}} > 0$, of the target. This corresponds to the nominal control objective
\begin{equation}
    \lim_{t\rightarrow\infty} \Vert \p(t) -\p_\mathrm{t} \Vert \leq d_{\mathrm{acc}}.
\end{equation}
The vehicle is expected to navigate in an obstacle-filled environment  whilst executing the nominal task. 
Let the position and velocity of an obstacle be defined as $ \p_i,  \v_i \in \mathbb{R}^2$, where $i \in \mathcal{I}_\mathrm{ob} \subset \mathbb{Z}$ represents the index of this obstacle. The control objective corresponding to this task consists of keeping a large enough distance to each obstacle:
\begin{equation}
   \Vert \p - \p_i \Vert \geq d_{\min,i},\quad \forall i \in \mathcal{I}_\mathrm{ob},\quad\forall t\geq t_0,
\end{equation}
where $d_{\min, i} > 0$ specifies the minimum distance the vehicle must keep to obstacle $i$, and $t_0\geq 0$ is some initial time.
\begin{figure}[t]
\centering
\subfloat[CC][]{\includegraphics[width=0.46\linewidth]{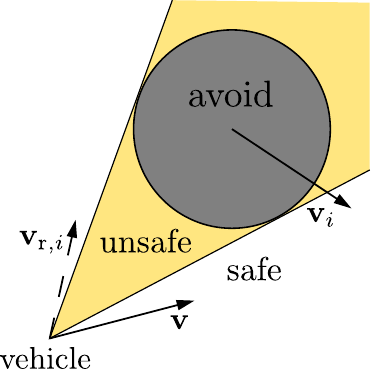}
\label{fig:cc}}\hfill
\subfloat[CC_geom][]{\includegraphics[width=0.44\linewidth]{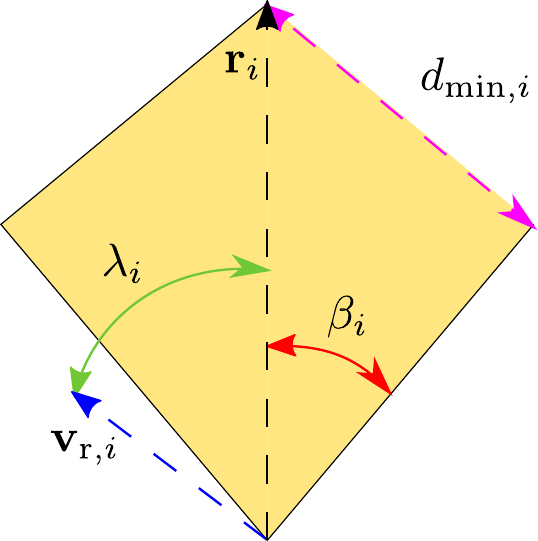}\label{fig:cc_geom}}
\caption{(a) demonstrates the collision cone principle; since the relative velocity is pointing outside of the cone, the vehicle is not headed towards a collision. (b) provides a geometric representation of the collision cone condition.}
\end{figure}

\subsection{The Velocity Obstacle}

The notion of velocity cones is particularly useful for predicting and addressing collisions between moving objects. Consider Figure~\ref{fig:cc}. The relative movement of the vehicle with respect to an obstacle is dictated by the (relative) velocity vector~$\vri \triangleq \v-\voi$. If this vector is directed within a distance $d_{\min, i}$ of the obstacle center, then the vehicle will move towards a collision, and eventually it will collide with the obstacle. This condition corresponds to a cone-shaped set in the velocity space, indicated in yellow, known as the \emph{collision cone}~\cite{Chakravarthy1998, FioriniShiller1998}. As illustrated in Figure~\ref{fig:cc_geom}, the half-angle of the cone can be found as
\begin{equation}\label{eq:beta_i}
    \beta_i \triangleq \arcsin\left(\frac{d_{\min, i}}{d_i}\right),
\end{equation}
where $d_i \triangleq  \Vert\ri \Vert$ and $\ri \triangleq \poi-\p$ is the vector starting at the vehicle origin, $\p$, connecting to the center of obstacle $i$. The requirement for a future collision can now be expressed
\begin{equation}\label{eq:col_cond}
    \lambda_i < \beta_i,
\end{equation}
where the angle between $\vri$ and $\ri$ is defined as
\begin{equation}\label{eq:lambda_i}
  \lambda_i \triangleq  \arccos\left(\frac{\vri \ri}{\Vert \vri \Vert d_i }\right).
\end{equation}
Similar to \cite[Lemma~1]{Haraldsen23} we can show that~\eqref{eq:col_cond} is a necessary condition for collision: 

\begin{lemma}\label{lem:S_C}
Given that the vehicle is not in a collision with obstacle $i$ at time $\tau\geq 0$, if, for all $\vri  \neq 0$, 
\begin{equation}\label{eq:safe_cond}
    \lambda_i(t) \geq \beta_i(t),\quad \forall t\geq \tau,
\end{equation}
then the vehicle will avoid a collision with the obstacle, i.e.,
\begin{equation}
 d_i(t) \geq  d_{\min,i},\quad \forall t \geq \tau.
\end{equation}
\end{lemma}
\begin{proof}
The time-derivative of the distance, $d_i$, is given by
\begin{equation}
\dot{d}_i=-\frac{\ri^\top}{d_i} \vri.
\end{equation}
Hence, if $\vri = 0$, a collision cannot occur since $\dot{d}_i = 0$. If $\vri \neq 0$, we use that \eqref{eq:safe_cond} is equivalent to
\begin{equation}
\cos \gamma_i \leq \cos \beta_i.
\end{equation}
With $\frac{\ri^\top}{d_i}\vri = \Vert \vri\Vert \cos \lambda_i$ and $\cos\beta_i =\frac{\sqrt{d_i^2-d_{\min,i}^2}}{d_i}$ by \eqref{eq:beta_i} and \eqref{eq:lambda_i}, we thus have 
\begin{equation}\label{eq:rnorm_dot}
    \dot{d}_i\geq  -\frac{\Vert \vri\Vert }{d_i} \sqrt{d_i^2-d_{\min,i}^2}.
\end{equation}
 Since $ \dot{d}_i \geq 0$ for $d_i = d_{\min,i}$ by~\eqref{eq:rnorm_dot}, it follows that $d_i(\tau)\geq d_{\min,i} \implies d_i(t)\geq d_{\min,i},\,\forall t\geq \tau$. 
\end{proof}

\begin{figure}[t]
\centering
\subfloat[VO][]{\includegraphics[width=0.46\linewidth]{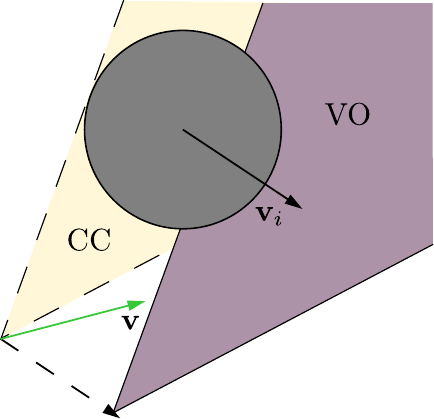}
\label{fig:vo}}\hfill
\subfloat[VO_geom][]{\includegraphics[width=0.44\linewidth]{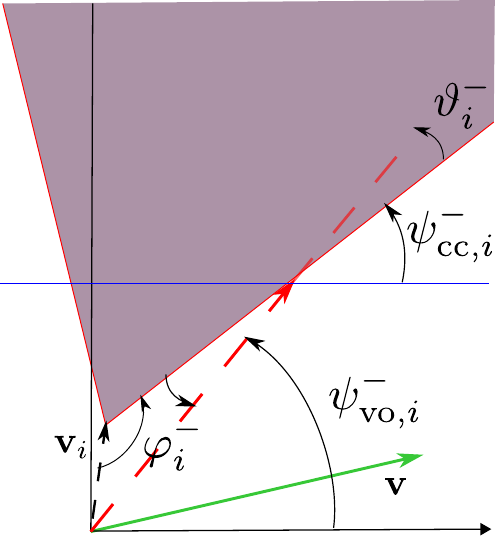}\label{fig:vo_geom}}
\caption{The velocity obstacle (VO) is obtained by translating the collision cone (CC) by the obstacle velocity in (a).  (b) illustrates the derivation of the intersection angle $\psi^-_{\mathrm{vo},i}$. }
\end{figure}

For generating avoidance maneuvers, this measure is not immediately useful. 
Moreover, 
 the analytical condition \eqref{eq:col_cond} becomes undefined when the relative velocity is zero. This situation can indeed occur in practical scenarios and may give rise to numerical challenges when using \eqref{eq:col_cond} as the basis for computing avoidance maneuvers as in~\cite{Tayal2023, Thontepu2023}.
The cone can more conveniently  be described in terms of the absolute vehicle velocity by adding the obstacle velocity to each velocity in the collision cone set. This transformation was proposed in~\cite{FioriniShiller1998} as the \emph{velocity obstacle} set, illustrated in Figure~\ref{fig:vo}.  According to Figure~\ref{fig:vo_geom}, the angles at which the vehicle velocity, $\v$, intersects the velocity obstacle cone are
\begin{equation}\label{eq:psivoi}
    \psi^\pm_{\mathrm{vo},i} = \psi_{\mathrm{cc},i}^\pm + \vartheta^\pm _i,
\end{equation}
where $\psi_{\mathrm{cc},i}^\pm \triangleq \angle \ri \pm \beta_i$ and
\begin{equation}\label{eq:vartheta}
    \vartheta^\pm _i \triangleq \arcsin\left(\frac{v_i}{v}\sin\left(\varphi_i^\pm\right)\right), \,\,\varphi_i^\pm\triangleq \pi-\psi_i + \psi_{\mathrm{cc},i}^\pm,
\end{equation}
with $v_i \triangleq \Vert \voi \Vert$ and $\psi_i \triangleq \angle \voi$ defined for conciseness. Interestingly, the angles~\eqref{eq:vartheta} can be undefined if the speed of the vehicle is below a certain magnitude dependent on the obstacle speed, which signifies that the vehicle's velocity vector does intersect the vertices of the velocity obstacle cone at any angle. Importantly, this shows that it may be necessary for the vehicle to increase its speed to maintain safety.

 While many vehicles exhibit limited or even complete absence of reversing capabilities, it is noteworthy that most collision avoidance methods formulated within the framework outlined in Section~\ref{sec:MathFramework} necessitate the presence of these abilities to ensure safety.  By leveraging the above introduced velocity obstacle concept, we formulate in the next section a strategy that instead allows the vehicle to move forward when attempting to avoid collisions and moreover stays defined for $\v_{\mathrm{r}, i} = 0$.

\section{Safety-Critical Control using\\ Velocity Obstacles}\label{sec:Controller}
We propose a control barrier function (CBF)-based strategy for computing the control inputs of the vehicle, consisting of the turning rate $u_1$ and forward acceleration $u_2$, aimed to preserve safety in dynamic environments. In Section~\ref{ssec:Nomcontrol}, we define the desired inputs of the vehicle for completing the nominal task of reaching a target.  Next, in Section~\ref{ssec:Speedadj}, we formulate a barrier function for adjusting the speed of the vehicle such that the directions~\eqref{eq:psivoi} corresponding to the velocity obstacle cone are well-defined. The required modification of the acceleration input takes place once an obstacle is within a certain proximity of the vehicle.  In Section~\ref{ssec:Oavoid}, we establish a barrier function for steering the vehicle in a collision-free direction, derived from the velocity obstacle cone. This controller is activated if the distance to the obstacle is even further reduced, indicating that a collision is approaching. We derive explicit conditions for safety under the above outlined strategy in Section~\ref{ssec:theor_safe}. 

\subsection{Nominal Control Laws}\label{ssec:Nomcontrol}
The vehicle is to move towards a target position at some specified forward speed $ v_{\mathrm{d}}  \in (0, v_{\max}]$. The desired heading is chosen as~\cite{breivik2005}
\begin{equation}
    \psi_{\mathrm{d}} = \atan(y_\mathrm{t}-y, x_\mathrm{t}-x),
\end{equation}
such that the vehicle takes the shortest path to $\p_\mathrm{t} \triangleq [x_\mathrm{t}, y_\mathrm{t}]^\top$. 
The desired acceleration and heading rate of the vehicle are chosen to keep the desired speed and heading as
\begin{equation}
    r_{\mathrm{d}} = -K_r (\psi-\psi_{\mathrm{d}}), \quad a_{\mathrm{d}}= -K_a (v-v_\mathrm{d}),
\end{equation}
where the angular difference is mapped to $(-\pi, \,\pi]$, and $K_r,K_a > 0$ are control gains.
\subsection{Speed Adjustments}\label{ssec:Speedadj}
Upon examination of \eqref{eq:vartheta}, it becomes evident that the vehicle must maintain a sufficiently high speed compared to that of the obstacle to enable the existence of an evasive direction.  The requirement can be summarized  as
\begin{equation}\label{eq:vcond}
    v \geq v_i|\sin(\varphi_i^j)|, \quad \forall j \in \{\pm\}.
\end{equation}
Based on this condition, we define a candidate CBF for adjusting the speed relative to obstacle $i$ as
\begin{equation}\label{eq:hvi}
    h_{v,i} = \min_{j,k } \, h^{k,j}_{v,i} - \kappa_{\min},
\end{equation}
where $j,k \in \{\pm\}$ and we define $h^{\pm,j}_{v,i} \triangleq v\pm v_i \sin(\varphi^j)$.  Moreover, since the angle~\eqref{eq:vartheta} is non-differentiable at $v = v_i \sin(\varphi_i^j)$, we add a margin $\kappa_{\min}>0$ ensuring that the derivative is well-defined when $h_{v,i} \geq 0$. The parameter is also useful for control design; by increasing $\kappa_{\min}$, it enforces a swifter evasion by effectively raising the speed kept relative to the obstacle speed.

While $h_{v,i}$ is not continuously differentiable, it is defined by the minimum of the smooth functions $h^{k,j}_{v,i}$. Hence, controller synthesis can be enabled by the framework proposed in~\cite{Glotfelter19} for the class of nonsmooth barrier functions~(NBFs). To reduce conservatism,  the speed is adjusted only for obstacles that are within a distance $d_{v,i} > 0$ of the vehicle. As such, the acceleration input is given by
\begin{equation}
    u_2 = \begin{cases}
        k_2 &\text{if } \mathcal{I}_{v, \mathrm{ob}} \neq \emptyset,\\
        a_{\mathrm{d}}& \text{otherwise},
    \end{cases}
\end{equation}
where $\mathcal{I}_{v, \mathrm{ob}} \triangleq \{i \in \mathcal{I}_{\mathrm{ob}}: d_{i} \leq d_{v,i}\}$ contains the indices of the obstacles that are closer than the required distances and $k_2$ represents the modified acceleration, computed via
\begin{equation}\label{eq:k2}
\begin{split}
   &k_2= \underset{a \in \mathcal{U}}{\arg \mathrm{min}} \,\,\Vert a-a_{\mathrm{d}}\Vert^2\\
    & \text{subj. to}  \,\,\dot{h}^{k,j}_{v,i} \geq -\alpha(h_{v,i}),\,\,  \forall(j,k) \in \mathcal{I}_{\epsilon_v,i},\,\forall i \in \mathcal{I}_{v, \mathrm{ob}}.
\end{split}
\end{equation}
Moreover, $\mathcal{I}_{\epsilon_v,i} \triangleq \{(j,k) : |h^{k,j}_{v,i} - h_{v,i}| \leq \epsilon_v\}$ is the almost-active set for the candidate control NBF (CNBF) of obstacle $i$, with $\epsilon_v > 0$. The constraint can be computed analytically
\begin{equation}\label{eq:h_v_kjdot}
\begin{split}
    \dot{h}^{\pm,j}_{v,i} = a \pm \dot{v}_{i} \sin(\varphi^j_i) \pm v_i\cos(\varphi^j_i)\left(-\Dot{\psi}_i + \Dot{\psi}_{\mathrm{cc},i}^j\right),
\end{split}
\end{equation}
where the expressions for $\dot{\psi}_{\mathrm{cc},i}^\pm$ are found in~\cite{Haraldsen23}.

\subsection{Obstacle Avoidance}\label{ssec:Oavoid}
If the distance to an obstacle is further reduced to a distance $d_{\psi,i}< d_{v,i}$, we modify the turning input $u_1$ with respect to the velocity obstacle cone through a second optimization.  Define the angular distances to the left and right vertex of the velocity obstacle cone, respectively, 
\begin{equation} \label{eq:delta_i}
    \delta_i^{\pm}\triangleq\pm \psi\mp \psi^\pm_{\mathrm{vo},i},
\end{equation}
where $\psi_{\mathrm{vo},i}^\pm$ are defined in~\eqref{eq:psivoi}. The angles are wrapped into the domain $\delta^\pm_i \in (-\pi,\pi]$ such that the distance
\begin{equation}\label{eq:deltavo}
    \delta_i^{\mathrm{vo}} \triangleq \underset{\delta_i \in \{\delta_i^+, \delta_i^-\}}{\arg\min}\, |\delta_i|,
\end{equation}
is negative when the vehicle velocity $\v$ lies within the velocity obstacle set and positive otherwise, see Figure~\ref{fig:delta}. Thus, the barrier function
\begin{equation}\label{eq:h_psi}
    h_{\psi,i} = \delta^{\mathrm{vo}}_i-\delta_{\min},
\end{equation}
encapsulates that the vehicle should keep a velocity outside of the velocity obstacle cone, with at least an angle~$\delta_{\min} \in [0,\frac{\pi}{2})$ to the closest vertex.  Note that $h_{\psi,i}$ is not continuously differentiable, however, the component functions  $h^\pm_{\psi,i} \triangleq \delta^\pm_i - \delta_{\min}$ are smooth.  Hence,  the desired heading rate is modified according to 
\begin{equation}
\begin{split}\label{eq:k1}
   &k_1= \underset{r\in \mathcal{U}}{\arg \mathrm{min}} \,\,\Vert r-r_{\mathrm{d}}\Vert^2\\
    & \text{subj. to}  \,\,\dot{h}^j_{\psi,i} \geq -\alpha(h_{\psi,i}),\,\, \forall j \in \mathcal{I}_{\epsilon_\psi,i},\,\, \forall i \in \mathcal{I}_{\psi, \mathrm{ob}},
\end{split}
\end{equation}
where $\mathcal{I}_{\epsilon_\psi,i} \triangleq \{j : |h^{j}_{\psi,i} - h_{\psi,i} |\leq \epsilon_\psi\}$ is the almost-active set for \eqref{eq:h_psi}, with $\epsilon_\psi> 0$, $\mathcal{I}_{\psi, \mathrm{ob}} \triangleq \{i \in \mathcal{I}_{\mathrm{ob}}: d_{i} \leq d_{\psi,i}\}$, and 
\begin{equation}\label{eq:h_psi_dot}
\begin{split}
    \dot{h}_{\psi,i}^\pm = \pm r \mp \dot{\psi}_{\mathrm{cc},i}^\pm\mp(-\dot{\psi}_i+\dot{\psi}_{\mathrm{cc},i}^\pm) \tfrac{  v_i \cos(\varphi_i^\pm)}{\sqrt{v^2-v_i^2\sin^2(\varphi_i^\pm)}} \\\mp(v\dot{v}_i-v_i u_2) \tfrac{ \sin(\varphi_i^\pm)}{v \sqrt{v^2 -v^2_i\sin^2(\varphi_i^\pm)}}.
\end{split}
\end{equation}
Notice that the acceleration input, $u_2$, is involved in the computation of $k_1$ through \eqref{eq:h_psi_dot}. Furthermore, the turning input is chosen as
\begin{equation}
    u_1 = \begin{cases}
        k_1 &\text{if } \mathcal{I}_{\psi, \mathrm{ob}} \neq \emptyset,\\
        r_{\mathrm{d}}& \text{otherwise}.
    \end{cases}
\end{equation}

\begin{figure}
    \centering
    \includegraphics[width=0.65\linewidth]{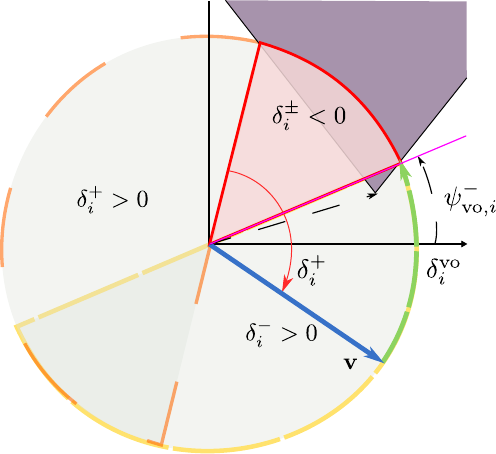}
    \caption{Intuition behind the angular distances $\delta^\pm_i$ and candidate CNBF $h_{\psi,i}$. The distance from the vehicle velocity $\v$ to the right vertex of the velocity obstacle cone, represented by $\delta^-_i$, is positive on the half-plane to the right of the corresponding direction ($\psi^-_{\mathrm{vo},i}$), and negative on the other half-plane. This property is reversed for the distance to the left vertex, $\delta^+_i$. The distance $\delta^{\mathrm{vo}}_i$ is then positive when $\v$ lies outside of the velocity obstacle, signifying safety, and negative otherwise. }
    \label{fig:delta}
\end{figure}

\subsection{Safety Guarantees}\label{ssec:theor_safe}
Since the relative degree is consistently defined, the functions $h_{\psi,i}$ and $h_{v,i}$ serve as valid CNBFs for~\eqref{eq:vehicle} with no input limitations (i.e., with $\mathcal{U} = \mathbb{R}^2$).  To extend this validity to them being CNBFs under input constraints,  the  behavior of an obstacle must be confined within similar physical bounds:
\begin{assumption}\label{ass:o_dyn}
    The dynamics of an obstacle are bounded:
    \begin{equation}
        |\dot{v}_i| \leq a_{\max,i}, \,|\dot{\psi}_i|\leq r_{\max,i}, \,v_i  \leq v_{\max,i} \quad \forall t \geq t_0,
    \end{equation}
    where $a_{\max,i}, r_{\max,i}, v_{\max,i}\geq 0$ are constant parameters.
\end{assumption}
 Given that every dynamic object is subject to certain physical constraints, Assumption~\ref{ass:o_dyn} is justified. In the final theorem, we outline sufficient conditions,  namely explicit lower bounds on the vehicle's speed and input limits, that under the proposed strategy guarantee forward invariance of the sets
\begin{equation}
   \!\! \mathcal{S}_{v,i} \triangleq \{\x \in \mathcal{D}\!:\! h_{v,i} \geq 0\},\quad\mathcal{S}_{\psi,i} \triangleq \{\x \in \mathcal{D}\!:\! h_{\psi,i} \geq 0\}.
\end{equation}

\begin{figure*}[t]
  \centering
  \subfloat[sim_multi_line_NE][]{\includegraphics[width=0.5\linewidth]{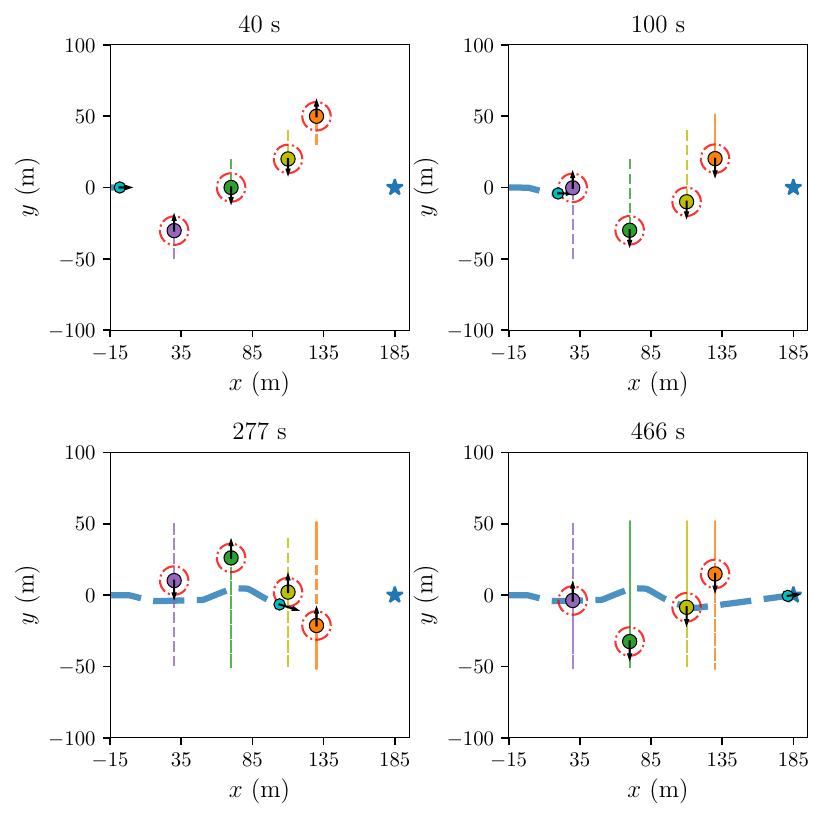}
\label{fig:sim_mol_NE}}
\subfloat[sim_multi_line_states][]{\includegraphics[width=0.5\linewidth]{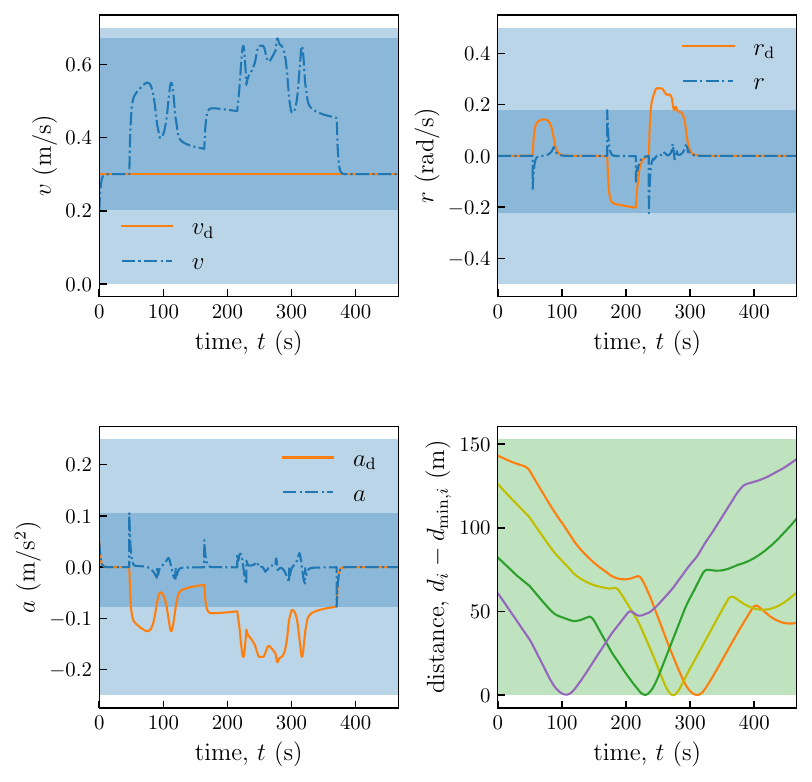}\label{fig:sim_mol_states}}
\caption{(a) shows the trajectory of the vehicle as it maneuvers between four obstacles to the target position, where the vehicle and the target are both marked in blue. The red, dashed circles indicate the distances $d_{\min,i}$ that should be kept between the center of the vehicle and the center of an obstacle. The corresponding forward speed, turning rate, and forward acceleration of the vehicle, and the distance between the vehicle and each obstacle, are displayed in (b).    }\label{fig:sim_mol}
\end{figure*}
\begin{figure*}
  \centering
\subfloat[sim_multi_circle_NE][]{\includegraphics[width=0.5\linewidth]{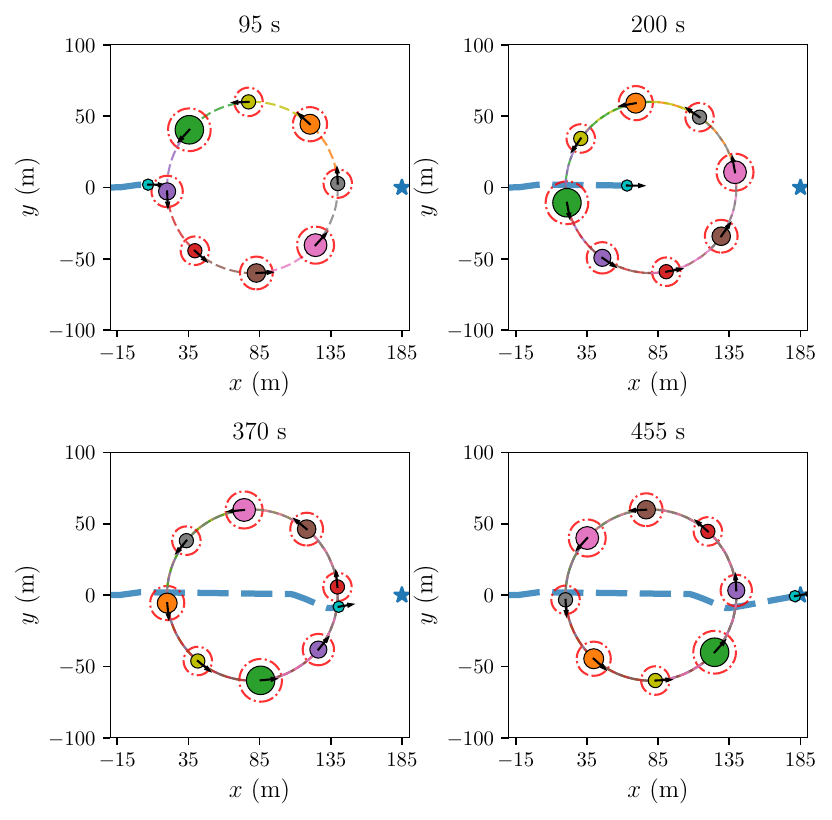}
\label{fig:sim_moc_NE}}
  \subfloat[sim_multi_circle_states][]{\includegraphics[width=0.5\linewidth]{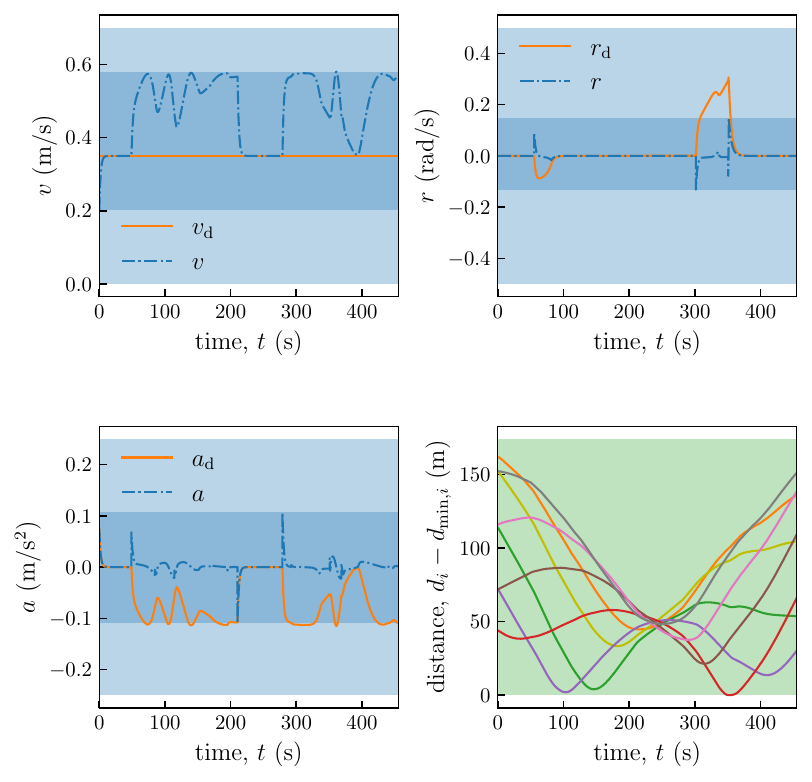}\label{fig:sim_moc_states}}
\caption{(a) shows the trajectory of the vehicle avoiding eight obstacles moving in a circle. The corresponding forward speed, turning rate, and forward acceleration of the vehicle, and the distances to the obstacles, are displayed in~(b)   }
\label{fig:sim_moc}
\end{figure*}
\begin{theorem}\label{thm:Theo_guarantees}
Consider a vehicle~\eqref{eq:vehicle} and let Assumption~\ref{ass:v} and \ref{ass:o_dyn} hold. If the maximum vehicle speed exceeds the maximum speed of obstacle $i$ such that $v_{\max} \geq v_{\max,i} +  \kappa_{\min}$ and 
    \begin{align}\label{eq:amax}
        a_{\max}& \geq a_{\max,i},\\ \label{eq:wmax}
        r_{\max}& \geq r_{\max,i} + \frac{a_{\max,i} + a_{\max}}{\kappa_{\min}},
    \end{align}
  then the control inputs \eqref{eq:k2} and \eqref{eq:k1} render the set $\mathcal{S}_{v,i} \cap \mathcal{S}_{\psi,i}$ forward invariant, i.e.,
    \begin{equation}
        \x(t_0) \in \mathcal{S}_{v,i} \cap \mathcal{S}_{\psi,i} \implies   \x(t) \in \mathcal{S}_{v,i} \cap \mathcal{S}_{\psi,i}, \forall t \geq t_0.
    \end{equation}
\end{theorem}
\begin{proof}
Consider first the acceleration input given by~\eqref{eq:k2}. It is evident from \eqref{eq:vcond} that the existence of a safe speed relative to obstacle $i$ requires that the vehicle's maximum speed is greater than or equal to the maximum speed of the obstacle. Denote $v' \triangleq v-\kappa_{\min}$. The property $\x \in \mathcal{S}_{v,i}$ remains satisfied whenever $v' \geq v_i$, in which case the vehicle's ability to remain in $\mathcal{S}_{v,i}$ is guaranteed if $a_{\max}$ exceeds the magnitude of $\dot{v}_i$.  For speeds $v' < v_i$, we evaluate $h_{v,i} \equiv h^{k,j}_{v,i}$ on the boundary  $\partial \mathcal{S}_{v,i} =\{\x \in \mathcal{D}: h^{k,j}_{v,i} = 0\}$, with $j,k \in \{\pm\}$ given by~\eqref{eq:hvi}. We then have $\varphi^j_i =k \arcsin(\frac{v'}{v_i})$, and
 \begin{equation}\label{eq:derivatives}
    v_i \dot{\varphi}^j_i \cos(\varphi^j_i)= k \left(a -  \dot{v}_i\frac{v'}{v_i}\right) , \quad \dot{v}_i \sin(\varphi_i^j) = k \dot{v}_i\frac{v'}{v_i}.
 \end{equation}
Insertion of~\eqref{eq:derivatives} in~\eqref{eq:h_v_kjdot} yields
 \begin{equation}
    \dot{h}^{k,j}_{v,i} =  a +  \dot{v}_i \frac{v'}{v_i} +\left( a -  \dot{v}_i\frac{v'}{v_i}\right) = 2a.
 \end{equation}
 Hence, $a\geq 0$ enforces $\dot{h}^{k,j}_{v,i}\geq 0$ in this case. Thus, forward invariance of $\mathcal{S}_{v,i}$ is assured under the condition \eqref{eq:amax}. 
 
Having established that the speed is lower bounded by $v \geq v_i |\sin(\varphi^j_i)|+\kappa_{\min}$ for all $j \in \{\pm\}$, we must now show that the vehicle can maintain $h_{\psi,i}\geq 0$ through the input~\eqref{eq:k1}. With $j$ now defined by~\eqref{eq:deltavo}, forward invariance of $S_{\psi,i}$ holds if $\dot{h}^j_{\psi,i} \geq 0$ on the boundary $\partial S_{\psi,i} = \{\x\in\mathcal{D}: h_{\psi,i}^j = 0\}$.  Within $S_{\psi,i}$, it holds that
 $\mp \dot{\psi}_{\mathrm{cc}}^\pm \geq 0$, which can be verified geometrically~\cite{Lalish2008,Haraldsen2020}. Thus, 
 \begin{equation}\label{eq:hpsidot_thm}
\dot{h}^\pm_{\psi,i} \geq \pm r\pm\dot{\psi}_i \mp\tfrac{ (v\dot{v}_i-v_i a) \sin(\varphi_i^\pm)}{v \sqrt{v^2 -v^2_i\sin^2(\varphi_i^\pm)}},
 \end{equation}
where we use that $\tfrac{v_i \cos(\varphi_i^\pm)}{\sqrt{v^2-v_i^2\sin^2(\varphi_i^\pm)}}\in [-1, 1]$. Hence, under Assumption~\ref{ass:o_dyn}, the condition~\eqref{eq:wmax} guarantees $\dot{h}_{\psi,i}^j \geq 0$, implying forward invariance of $\mathcal{S}_{v,i}\cap\mathcal{S}_{\psi,i}$.
\end{proof}
 Notably, when $\x \in \mathcal{S}_{v,i}\cap\mathcal{S}_{\psi,i}$,  collision avoidance of obstacle $i$ is guaranteed by Lemma~\ref{lem:S_C}. Under the conditions outlined in Theorem~\ref{thm:Theo_guarantees}, the proposed controller ensures the vehicle's safety in encounters with a single obstacle, given that the parameters $d_{v,i}$ and $d_{\psi,i}$ are sufficiently large to facilitate convergence to $S_{v,i} \cap S_{\psi,i}$ prior to any potential collision. In scenarios involving multiple obstacles, the vehicle's safety depends on whether the independent avoidance maneuvers are compatible, in which case safety is upheld.
\begin{remark}
  Note that  \eqref{eq:wmax} is a conservative bound, as the forward acceleration typically opposes the obstacle's acceleration, resulting in a reduction of the final term in \eqref{eq:hpsidot_thm}.  Hence, the parameter $\kappa_{\min}$ can be selected considerably lower than what is required by this condition. 
\end{remark}

\section{Simulations}\label{sec:Simulations}
Now, we demonstrate the proposed collision avoidance strategy in simulations of a vehicle with the unicycle kinematics~\eqref{eq:vehicle}, navigating among moving obstacles to a target position~$\p_\mathrm{t} = [185, 0]^\top$~m. We used $r_{\max} = 0.5$~rad/s and $a_{\max} = 0.25$~m/s$^2$ as input constraints and a maximum forward speed of $v_{\max} = 0.7$~m/s. The radius of the vehicle was selected as $R = 5$~m giving a minimum distance between the centers of the vehicle and obstacle $i$ as $d_{\min,i} = R + R_i$,  $R_i > 0$  being the radius of the obstacle. In practice, this parameter can be increased above the combined radii to keep some space between the object boundaries. Furthermore, we employed distances $d_{\psi,i} = 30 + d_{\min,i}$~m and $d_{v,i} = d_{\psi,i} + 5$~m. In the computation of $h_{v,i}$ and $h_{\psi,i}$,  we used the parameters $\epsilon_\psi = \epsilon_v= \kappa_{\min} = \delta_{\min}= 0.05$. We used a linear class-$\mathcal{K}$ function $\alpha(h) = -\gamma h, \gamma = 0.5$  in both cases. The nominal control gains were set to $K_a = K_r = 0.5$ and the acceptance distance to $d_{\mathrm{acc}} = 4$~m. The first-order Euler method was used for the numerical integration with step-size~$0.01$~s.

First, the vehicle must navigate past four obstacles\footnote{See a simulation video at: \url{https://youtu.be/phQpTdAGmig}.}, each with a radius, $R_i$, of $5$~m. The desired speed was fixed at $v_\mathrm{d} = 0.3$~m/s. The obstacles all move parallel to the inertial $y$-axis at a speed of $v_i = 0.5$~m/s, going back and forth between $\pm 50$~m (the speeds are changed at an acceleration of $0.1$~m/s$^2$ at the turning points). The vehicle must thus pass all four obstacles to reach the goal; the initial configuration of the obstacles is shown in the top-left of Figure~\ref{fig:sim_mol_NE}.  The vehicle can be seen to move towards the closest obstacle, triggering an increase of the speed, viewed in the top-left plot of Figure~\ref{fig:sim_mol_states}. This is attributed to the fact that the obstacle keeps a higher speed than the vehicle, and to satisfy~\eqref{eq:vcond} the vehicle must raise the speed correspondingly. From the top-right plot, we observe that the safety-critical controller~\eqref{eq:k1} is active as there is a temporary change of direction, during which the vehicle passes behind the obstacle. The same behaviour is displayed as the vehicle encounters more obstacles. Note that, since the barrier functions $h_{\psi,i}$ are not consistently defined, they have been omitted from the plots. Moreover, collision avoidance can be verified from the distances to each obstacle which are kept above the required minimums at all times. 

Next, eight obstacles of varying radii (ranging from $5$ to $10$~m) must be avoided, see Figure~\ref{fig:sim_moc_NE}. They are simulated to move synchronously in a clockwise circle about the point $[80, 0]^\top$~m, with speeds $v_i = 0.525$~m/s and  turning rates of $0.00875$~rad/s. The nominal vehicle speed was selected as $v_\mathrm{d} = 0.35$~m/s. From Figure~\ref{fig:sim_moc_states}, we observe that the speed controller raises the forward speed to an appropriate level as the vehicle gets closer to the obstacles. The speed maneuver is accompanied by a slight change of direction, enabling the vehicle to pass between two obstacles. The vehicle's speed returns to the nominal speed before a new maneuver is generated, taking the vehicle safely behind an obstacle, verified by the bottom-right plot of Figure~\ref{fig:sim_moc_states}. The vehicle reaches the target position shortly after. Importantly, safety of the vehicle is preserved in both cases without the necessity of stopping or reversing, while the generated inputs and the maximum vehicle speed stay within expected bounds.

\section{Conclusions and Future Work}\label{sec:Conclusion}

This paper proposed an obstacle avoidance strategy for vehicles with first-order nonholonomic constraints, moving in dynamic environments. We proposed to regulate the vehicle speed and orientation separately via two control barrier functions (CBFs), with the respective safety conditions derived from the velocity obstacle principle.  Unlike most other approaches defined within the CBF framework, the proposed strategy does not require the vehicle to brake or move backwards to avoid collisions. On the contrary, it enables the vehicle to maintain a nonzero forward speed by only adjusting the speed once it falls below the required level for avoiding an obstacle. This can benefit the vehicle's progression towards a goal and accommodates systems with operational requirements that prevent them from reversing.  Meanwhile, the steering controller diligently enforces the necessary turning maneuvers to avoid potential collisions. We gave theoretical assurance of safety under explicit conditions. The resulting performance was showcased through simulations of challenging scenarios with multiple moving obstacles, displaying adept and secure navigation through obstacle-filled environments.

In the future, we aim to implement the strategy in experiments on a marine vessel and explore more strategic utilization of the speed to further enhance the performance.

\bibliographystyle{IEEEtran}
\bibliography{references}
\end{document}